\RequirePackage{ifpdf}
\ifpdf
\documentclass[11pt,pdftex]{article}
\else
\documentclass[11pt,dvips]{article}
\fi

 \usepackage{epsfig}
 \usepackage{color,graphics}
 \usepackage{latexsym,amsmath,amsfonts,amsthm,amssymb}
 \usepackage{euscript,amsbsy}
 \usepackage{graphicx}
 \usepackage{caption}

\usepackage{latexsym}
\usepackage{amscd}
\usepackage{color,graphics}
\usepackage{graphicx}

\usepackage{ifpdf}   
\DeclareGraphicsExtensions{.pdf,.png,.jpg,.mps}

\input colordvi    

\numberwithin{equation}{section}

\newcommand{\sign}{{\mathrm{sgn}}}

\newtheorem{theo}{Theorem}[section]
\newtheorem{prop}[theo]{Proposition}
\newtheorem{lem}[theo]{Lemma}
\newtheorem{cor}[theo]{Corollary}

\newtheorem{defi}[theo]{Definition}
\newtheorem{rem}[theo]{Remark}
\newtheorem{ass}[theo]{Assumption}
\newtheorem{conj}[theo]{Conjecture}

\newcommand{\ba}{\begin{array}}
\newcommand{\ea}{\end{array}}
\newcommand{\bea}{\begin{eqnarray}}
\newcommand{\eea}{\end{eqnarray}}
\newcommand{\bead}{\begin{eqnarray*}}
\newcommand{\eead}{\end{eqnarray*}}
\newcommand{\be}{\begin{equation}}
\newcommand{\ee}{\end{equation}}
\newcommand{\bed}{\begin{displaymath}}
\newcommand{\eed}{\end{displaymath}}

\newcommand{\bco}{\begin{conj}}
\newcommand{\eco}{\end{conj}}

\newcommand{\bl}{\begin{lem}}
\newcommand{\el}{\end{lem}}
\newcommand{\bp}{\begin{prop}}
\newcommand{\ep}{\end{prop}}
\newcommand{\bt}{\begin{theo}}
\newcommand{\et}{\end{theo}}
\newcommand{\bpr}{\begin{proof}}
\newcommand{\epr}{\end{proof}}

\newcommand{\bc}{\begin{cor}}
\newcommand{\ec}{\end{cor}}

\newcommand{\br}{\begin{rem}}
\newcommand{\er}{\end{rem}}
\newcommand{\bd}{\begin{defi}}
\newcommand{\ed}{\end{defi}}
\newcommand{\bass}{\begin{ass}}
\newcommand{\eass}{\end{ass}}

\begin{document}

\title {A continuum limit for the Kronig-Penney model}
\author{Matteo Colangeli$^{1}$, Sokol Ndreca$^{2}$, Aldo Procacci$^{1}$ \\\\
\footnotesize{$^1$Dep. Matem\'atica-ICEx, UFMG, CP 702 Belo Horizonte - MG, 30161-970 Brazil}\\
\footnotesize{$^2$Dep. Estatistica-ICEx, UFMG, CP 702 Belo Horizonte - MG, 30161-970 Brazil}\\
\small{emails: {colangeli@mat.ufmg.br};~{sokol@est.ufmg.br};~{aldo@mat.ufmg.br}}\\}
\maketitle
\begin{abstract}\noindent
We investigate the transmission properties of a quantum one-dimensional periodic system of fixed length $L$, with $N$ barriers of constant height $V$ and width $\lambda$,
and $N$ wells of width $\delta$. In particular, we  study the behaviour of the transmission coefficient in the limit $N\to \infty$,  with $L$ fixed. This is achieved
by  letting  $\delta$ and $\lambda$
both scale as $1/N$, in such {a} way that their ratio $\gamma= \lambda/\delta$ is a fixed parameter  characterizing the model.
In this continuum limit the multi-barrier system behaves as it were constituted by a unique barrier
of constant height $E_o=(\gamma V)/(1+\gamma)$. The analysis of the dispersion relation of the model shows the presence of forbidden energy bands at any finite $N$.
\end{abstract}


\section{Introduction}

The Kronig-Penney (K-P) model is one of the few solvable models in quantum mechanics which makes it possible to investigate the properties of electronic transport in real solids.\\
In their seminal paper \cite{kp}, R. de L. Kronig and W. G. Penney discussed the dispersion relation characterizing the transmission of an electron through a periodic potential in a one-dimensional domain. They were able to unveil the ``opening'', at the edges of the Brillouin zones, of the continuous quadratic curve typical of a free particle,
{thus marking} the onset of forbidden energy bands.
The relevance of the K-P model is two-fold: on the one hand, being a model amenable to an analytical solution, it enables {to study the occurrence} of forbidden bands {via} the Bloch Theorem  \cite{bloch}. On the other hand, the model proves useful to {highlight} the role of the periodic potential in determining whether the system {carries} the properties of an assembly of isolated wells, each equipped with a discrete energy spectrum,
or, rather, as a superlattice, characterized by continuous energy levels, possibly separated by forbidden bands.
In other words, by tuning the strength of the  potential in the K-P model, one is able to recover the wide range of conduction properties of real solids,
i.e. conductors, semiconductors or insulators, depending on the value of the Fermi energy in the resulting band structure. \\
The original K-P model postulates a periodic system of infinite size in which the electron-phonon interaction is disregarded, so that boundary and dissipative
effects in the bulk can be neglected.
Thus, the chance of interpreting some of the properties of real solids by means of a simple, idealized, model encouraged a vast literature on the transport
theory of quantum {multi-barrier} systems \cite{cho,esp}
and paved the way to the growing field of mesoscopic physics \cite{esatsu,esatsu2,rau}. Moreover, the K-P model is also a source of inspiration for the modeling, through the prism of solid state physics, of some promising, recently engineered, materials, e.g. the graphene  \cite{park,gatt,masir}.\\
In this work, we consider a variant of the standard K-P model consisting of an array of periodic cells with finite total length.
We then investigate the electronic transport properties of the periodic multi-barrier system in the \textit{continuum limit}, i.e., we let the number of cells diverge, and simultaneously rescale the size of each cell, so as to keep the total length of the sample fixed. Such procedure leads to
a limiting behaviour of the K-P model which differs{, in general,}  from the {thermodynamic} limit discussed in the literature (see e.g. Ref. \cite{kam}).\\
Thus, the investigation of the continuum limit of the K-P model {makes it possible} to explore the mathematical properties of a paradigmatic quantum mechanical system
under a different perspective and {also clarify} the physical implications of the invoked limiting procedure.
At the same time, our investigation also points towards the development of an ``effective''
theory for the low-dimensional samples of interest in the modern semiconductor technology, characterized
by a finite length and made of a {typically} large number of {layers}.\\
The paper is organized as follows.\\
In Sec. \ref{sec:sec1} we illustrate the general features of the K-P model under investigation.\\
In Sec. \ref{sec:sec2} we review the formalism, based on the transfer matrix technique, which makes it possible to determine, for an arbitrary number of barriers, the analytical expression of the transmission coefficient.\\
In Sec. \ref{sec:sec3} we provide the definition of the continuum limit of the K-P model and state our main result, concerning the asymptotic properties of the transmission coefficient.\\
In Sec. \ref{sec:sec4} we derive, via the Bloch Theorem, the dispersion relation of our finite K-P model, and discuss the appearance of energy bands, by also comparing with the results known for the original K-P model.\\
Conclusions are drawn in the final Sec. \ref{sec:concl}.

\section{The model}
\label{sec:sec1}

We consider a periodic one-dimensional system made of $N$ cells on a lattice, each cell consisting of one barrier
and one well with lengths denoted, respectively, by $\lambda$ and $\delta$. The total length of the system is $L=N p$, with $N\in \mathbb{N}$,
 where $p=\lambda+\delta$ denotes the period of the lattice. {As shown in Fig. \ref{barriers}, the barriers have constant height {$V$}
 and are delimited by a set of $2N$ points, denoted by $x_0=0,...,x_{2N-1}=L-\delta$, hereafter called \textit{nodes}.
\begin{figure}
   \begin{center}
   \includegraphics[width=0.7\textwidth]{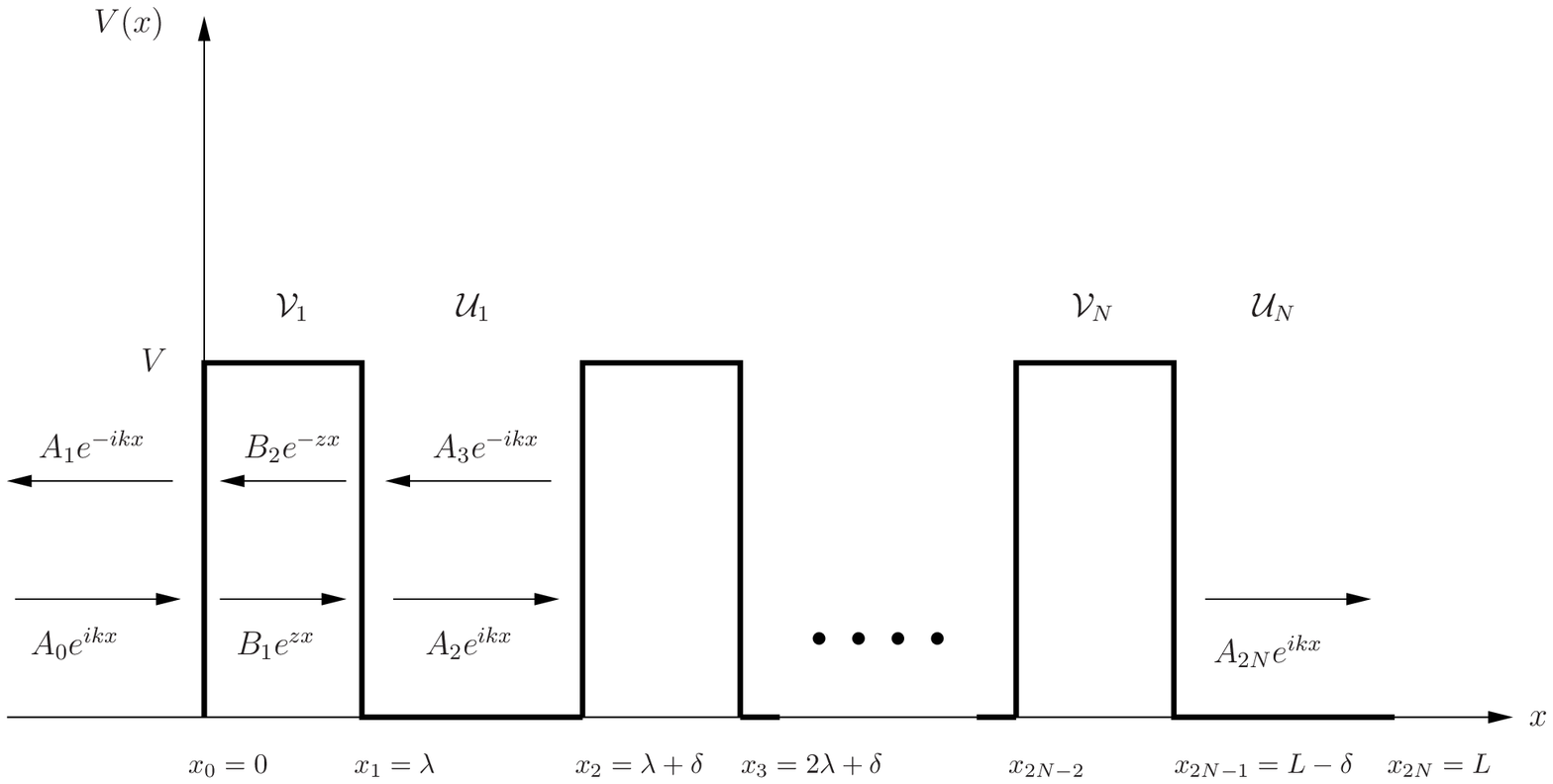}
   \caption{{Kronig-Penney model with finite total length $L$}, consisting of $N$ potential barriers with height $V$ and width $\lambda${, and $N$ wells} with width $\delta$.}\label{barriers}
   \end{center}
\end{figure}
We also denote, for $\ell=1,2,\dots, N$, $\mathcal{U}_\ell= [x_{2\ell-1}, x_{2\ell}]$ and $\mathcal{V}_\ell= (x_{2\ell-2}, x_{2\ell-1})$.\\
This model is described by a wavefunction $\Psi(x)$, obeying the time independent Schr\"{o}dinger Equation
\be
\left[-\sigma\frac{d^2}{dx^2}+V(x)\right]\Psi(x)=E\Psi(x) \quad \text{for $x\in[0,L]$} \label{se}
\ee
with $\sigma=\hbar^2/(2m)$, where $\hbar$ is the Planck constant and $m$, $E$ denote, respectively, the rest mass and the energy of the electron. The potential $V(x)$, in Eq. \eqref{se}, is defined as
\be
V(x)=V~~~ \text{for $x\in \mathcal{V}_\ell$} \quad, \quad V(x)=0 ~~~ \text{for $\{x\in \mathcal{U}_\ell$, ~ $x\le0$,~ $x> L$\}}
\label{potent}
\ee
We measure lengths in nanometers and the energy in units of $0.038$ electronvolts, in such a way that  hereafter  $\sigma=1$.} \\
Then the solution of Eq. (\ref{se}) takes the form:

\be
\Psi(x)=\left\{\begin{array}{cccc}
A_0e^{ikx}+A_1e^{-ikx} & \hbox{if $x\le 0$}\\ \\
A_{2N} e^{ikx} &  \hbox{if $x> L$}\\ \\
A_{2\ell} e^{ikx}+A_{2\ell+1} e^{-ikx} &  \hbox{if $x\in \mathcal{U}_\ell$} \\ \\
B_{2\ell-1} e^{zx}+B_{2\ell} e^{-zx} &  \hbox{if $x\in \mathcal{V}_\ell$ }
               \end{array}\right. \label{psi}
\ee with \be k=\sqrt{{E}} \quad \text{and} \quad z=\sqrt{{V-E}}
\label{ener} \quad \ee \noindent The boundary conditions prescribe
$A_0>0$ for the amplitude of the plane wave entering from the left
boundary and $A_{2N+1}=0$ (no wave enters or is reflected from the
right boundary).
One may define the transmission coefficient $S_N$
for the multi-barrier system as \be S_N=\frac{|A_{2N}
|^2}{|A_{0}|^2} \label{J} \ee
In the next Section  we present an explicit derivation of $S_N$ in terms of the Chebyshev polynomials of the second kind. \\
To this aim, we will also introduce the basic notation used in the paper.

\section{General results}
\label{sec:sec2}

\noindent
Let us introduce the notations
\be
\mathbf{\Delta}[\alpha x]=\left( \begin{array}{c c}
e^{\alpha x} & 0 \\
0 &  e^{-\alpha x}
\end{array} \right) \label{not0}
\ee
and
\be
\mathbf{T}[\alpha,x]=\left( \begin{array}{c c}
e^{\alpha x} & e^{-\alpha x} \\
\alpha e^{\alpha x} & -\alpha e^{-\alpha x}
\end{array} \right)=\mathbf{T}[\alpha,0]\mathbf{\Delta}[\alpha x] \label{not}
\ee
Using the transfer matrix method, see e.g. Ref. \cite{col2013}, one may relate the amplitudes $A_0,A_1$ of the incoming wave corresponding for $x\le 0$
to the amplitude $A_{2N}$ of the wave outgoing the $N$-th barrier.
Indeed, by imposing the standard conditions of continuity of the wavefunction and its first derivative at the nodes, one finds, after some algebra

\be
\left( \begin{array}{c}
A_{0} \\
A_{1}
\end{array} \right)=\prod_{j=1}^N \mathbf{Q}_j \left( \begin{array}{c}
A_{2N} \\
0
\end{array}\right)  \label{tranmat}
\ee
where the $j$-th transfer matrix $\mathbf{Q}_j$ belongs to $SL(2,\mathbb{C})$ and is defined as
\be
\mathbf{Q}_j=\mathbf{T}^{-1}[ik,x_{2j-2}]\mathbf{T}[z,x_{2j-2}]\mathbf{T}^{-1}[z,x_{2j-1}]\mathbf{T}[ik,x_{2j-1}] \quad  \text{for} \quad j=1, 2, \dots, N
\ee
Hence, using Eq. (\ref{not}), one may rewrite Eq. (\ref{tranmat}) in the following form \cite{leo2011}
\be
\left( \begin{array}{c}
A_{0} \\
A_{1}
\end{array} \right)=\mathbf{\Delta}[ik\delta] \mathbf{M}^N\mathbf{\Delta}[ik(L-\delta)] \left( \begin{array}{c}
A_{2N} \\
0
\end{array}\right)  \label{tranmat2}
\ee
where
\be
 \mathbf{M}=\mathbf{\Delta}[-ik\delta]\mathbf{T}^{-1}[ik,0]\mathbf{T}[z,0]\mathbf{\Delta}[-z\lambda]\mathbf{T}^{-1}[z,0]\mathbf{T}[ik,0] \quad \label{matrix1}
\ee
Note that, using the vector notation
$$
\Psi(x)=\mathbf{\Delta}[ikx] \left( \begin{array}{c}
A_{2\ell} \\
A_{2\ell+1}
\end{array}\right) \quad \text{for $x\in \mathcal{U}_\ell$}
$$
and
$$
\Psi(x)=\mathbf{\Delta}[ikx] \left( \begin{array}{c}
A_{0} \\
A_{1}
\end{array}\right) \quad \text{for $x\le 0$}
$$
formula \eqref{tranmat2} implies that
\be
{\Psi(x)\big\vert_{x=L}=\left(\mathbf{\Delta}[ik\delta]\mathbf{M}^{-1}\mathbf{\Delta}^{-1}[ik\delta]\right)^{N}\Psi(x)\vert_{x=0}}\label{psiper}
\ee
From {Eqs.} \eqref{J} and (\ref{tranmat2}), one obtains
\be\label{newJ} S_N= \frac{1}{|(\mathbf{M}^N)_{11}|^2} \ee \noindent
A straightforward calculation shows that the entries $m_{ij}$  of the $2\times 2$  matrix
$\mathbf{M}$ are given by
\bea
\Re{(m_{11})}&=& \cos(k\delta) \cosh(z\lambda) + \frac{z^2 - k^2}{2 k z} \sin(k\delta) \sinh(z\lambda)   \nonumber\\
\Im{(m_{11})}&=& -\sin(k\delta) \cosh(z\lambda) + \frac{z^2-k^2}{2 k z} \cos(k\delta) \sinh(z\lambda) \nonumber\\
\Re{(m_{12})}&=& \frac{V}{2 k z} \sin(k\delta) \sinh(z\lambda)   \nonumber\\
\Im{(m_{12})}&=&\frac{V}{2 k z} \cos(k\delta) \sinh(z\lambda)  \label{entries}
\eea
{and}
\be\label{entries2}
m_{22}= m_{11}^*  \quad  \text{and} \quad m_{21}= m_{12}^*  \quad
\ee
\noindent
In the sequel, we will make use of the shorthand notation $\Phi=\Re{(m_{11})}$, so that
\be\label{phi0}
\Phi= \cos(k\delta) \cosh(z\lambda) + \frac{z^2 - k^2}{2 k z} \sin(k\delta) \sinh(z\lambda)
\ee
Denoting the eigenvalues of $\mathbf{M}$ by $\mu_1$ and $\mu_2$, one finds
\be \label{mu}
\mu_1=  \Phi - \sqrt{\Phi^2- 1} \quad \text{and} \quad
\mu_2=\Phi + \sqrt{\Phi^2- 1}
\ee
Note that
\be\label{mumu}
\mu_1 \mu_2 =1
\ee
\noindent
i.e.  $\mathbf{M}$ is an element of $SL(2,\mathbb{C})$.
Moreover, depending on the value of $\Phi$, the eigenvalues $\mu_1$ and $\mu_2$ can be real or complex-valued.\\
Let us now provide the general expression of the transmission coefficient $S_N$ in terms of the Chebyshev polynomials of the second kind.
\begin{lem} \label{importante}
Let us consider the model ruled by Eq. (\ref{se}). Then,
for any $N\in\mathbb{N}$ and for any $\delta,\lambda,V \in \mathbb{R^{+}}$, the transmission coefficient $S_N$ attains the structure
\be
{S_N
= \left[1+\left(\frac{V}{2kz}\sinh(z\lambda)U_{N-1}(\Phi)\right)^2\right]^{-1}}  \label{SKP}
\ee
where $U_{N}(\Phi)$ are the Chebyshev polynomials of the second kind in the variable $\Phi$.
\end{lem}
\begin{proof}
We use a general formula for the $n$-th power of a $2\times2$ matrix, (see Ref. \cite{will}).
Given a $2\times2$ matrix $\mathbf{M}$, with (possibly coincident) eigenvalues $\mu_1$ and $\mu_2$, and denoting by $\mathbf{I}$ the $2\times2$ identity matrix, it holds
\begin{equation}
\mathbf{M}^N=\frac{\mu_1^N - \mu_2^N}{\mu_1 - \mu_2}\mathbf{M} -
 \frac{\mu_2\mu_1^N - \mu_1\mu_2^N}{\mu_1 - \mu_2}\mathbf{I}  \label{theo}
\end{equation}
for $\mu_1\neq\mu_2$, whereas, if $\mu=\mu_1=\mu_2$, it holds
\begin{equation}
\mathbf{M}^N=N \mu^ {N-1}\mathbf{M} -
 (N-1)\mu^N\mathbf{I} \label{theo2}
\end{equation}
\noindent
From Eqs. \eqref{mumu}, (\ref{theo}) and (\ref{theo2}) one can write $\mathbf{M}^N$ in Eq. (\ref{tranmat2}) as
\be
\mathbf{M}^N=U_{N-1}(\Phi)\mathbf{M}-U_{N-2}(\Phi)\mathbf{I} \label{Mn}
\ee
with
\be
U_{N-2}(\Phi) = \frac{\mu_1^{N-1} - \mu_2^{N-1}}{\mu_1 - \mu_2} \quad \text{and} \quad U_{N-1} (\Phi)= \frac{\mu_1^{N} - \mu_2^{N}}{\mu_1 - \mu_2}~ \label{cheb}
\ee
for distinct eigenvalues, or
\be
U_{N-2}(\Phi) = (N-1)\mu^N \quad \text{and} \quad U_{N-1} (\Phi)=N \mu^ {N-1}  \label{chebbis}
\ee
for coincident eigenvalues.\\
Note that
\be U_{-1}(\Phi)=0 \quad \text{and} \quad U_0(\Phi)=1 \label{incon}
\ee
From \eqref{mu} and Eqs. (\ref{cheb}), (\ref{chebbis}) one obtains the following recurrence relation
\be
U_N(\Phi)=2\hspace{0.7mm}\Phi\hspace{0.7mm}U_{N-1}(\Phi)-U_{N-2}(\Phi) \label{recur}
\ee
Equation (\ref{recur}), supplemented by the initial conditions (\ref{incon}), allows one to identify the functions $U_N(\Phi):\mathbb{R}\rightarrow\mathbb{R}$ with the Chebyschev polynomials of the second kind \cite{abram}.
It is easy to see, in particular, that the first entry $(\mathbf{M}^N)_{11}$ of the matrix $\mathbf{M}^N$ reads
\be
\left(\mathbf{M}^N\right)_{11} =\left[\Re{(m_{11})}+i \Im{(m_{11})}\right] U_{N-1}(\Phi)
-U_{N-2}(\Phi) \label{entrata}
\ee
Thus, using (\ref{recur}), one arrives at
\be
\left(\mathbf{M}^N\right)_{11}=
U_N(\Phi)-m_{22}\hspace{0.7mm}U_{N-1}(\Phi) \label{m11}
\ee
Similarly, one can also show that
\be
\left(\mathbf{M}^N\right)_{22}=
U_N(\Phi) - m_{11}\hspace{0.7mm}U_{N-1}(\Phi) \label{m22}
\ee
Therefore, the matrix $\mathbf{M}^N$ has the structure:
\be
\mathbf{M}^N = \left(
\begin{array}{ccc}\label{MN}
U_N-m_{22}\hspace{0.7mm}U_{N-1}& U_{N-1}\hspace{0.7mm}m_{12}\\
U_{N-1}\hspace{0.7mm}m_{21} & U_N-m_{11}\hspace{0.7mm}U_{N-1} \\
\end{array} \right)
\ee
\noindent
One can readily check that $\left(\mathbf{M}^N\right)_{11}=\left[\left(\mathbf{M}^N\right)_{22}\right]^*$
and $\left(\mathbf{M}^N\right)_{12}=\left[\left(\mathbf{M}^N\right)_{21}\right]^*$.\\
Moreover, since $\mathbf{M} \in SL(2,\mathbb{C})$, one also has  $\mathbf{M}^N \in SL(2,\mathbb{C})$. Hence, one finds
\be\label{mn}
\left | \left(\mathbf{M}^N\right)_{11}\right|^2=1+\left | \left(\mathbf{M}^N\right)_{12}\right|^2
\ee
and the proof follows by using Eqs. (\ref{newJ}), \eqref{entries} and \eqref{mn}.\\
\end{proof}
\br
Note that, by computing directly  $\det(\mathbf{M}^N)$ from \eqref{MN}, it must hold that
$$
U_{N-1}^2(\Phi)-U_{N}(\Phi)U_{N-2}(\Phi)=1
$$
which is a well known property of the Chebyshev polynomials of the second kind \cite{abram}.
\er
\noindent

\section{The continuum limit of the model}
\label{sec:sec3}

We now proceed with the investigation of the continuum limit of the finite K-P model. \\
Let us start providing the following
\begin{defi}
\label{continuum}
Let $\gamma,L \in \mathbb{R}^{+}$ be fixed.
The $(\gamma,L)$-continuum limit of the K-P model is found by taking{, in Eq. \eqref{se},} the limits $N\rightarrow \infty$, $\delta\rightarrow 0$ and $\lambda\rightarrow 0$, in such a way that $\lambda/\delta = \gamma$ and $N(\lambda+\delta)=L$.
\end{defi}
\noindent
Thus, differently from the standard K-P model, we let the quantities $\delta$, $\lambda$ and $p$ depend on $N$, and replace them, correspondingly, with the symbols $\delta_N$, $\lambda_N$ and $p_N$.
Similarly, we also replace $\Phi$ with $\Phi_N$, to take into account the dependence of $\Phi$ from $N$, via the explicit dependence from $\delta$ and $\lambda$, cf. Eq. \eqref{phi0}. \\
Clearly we have the relations
\be
\delta_N=\frac{1}{1+\gamma}\frac{L}{N} \quad \text{and} \quad \lambda_N=\frac{\gamma}{1+\gamma}\frac{L}{N} \label{def}
\ee
Let us also define, for later convenience,
\be
E_o=\frac{\gamma}{1+\gamma}V \label{Eo} \quad
\ee
We are now ready to state the main result of this paper, which provides the {expression of the transmission coefficient $S_N$ in (\ref{SKP}) in the continuum limit}, to be denoted by $\bar{S}$.
\begin{theo}\label{teo}
Let us consider the finite K-P model described by the time independent Schr\"{o}dinger Equation (\ref{se}).
Then, according to the Definition \ref{continuum}, the {continuum limit} of the transmission coefficient $S_N$ is given by
\be
{\bar{S}=\lim_{N\rightarrow\infty}S_N=\left[1+\frac{E_o^2}{4E}
\left[\frac{\sin\left({L\sqrt{E-E_o}}\;\right)}{\sqrt{E-E_o}}\right]^2
\right]^{-1}}
\label{SKP2}
\ee
\end{theo}
\begin{proof}
In order to plug the expressions (\ref{def}) in (\ref{SKP}),  observe that it holds
\be
\sinh\left(\frac{z\gamma L}{1+\gamma}\frac{1}{N}\right)=\frac{z\gamma L}{1+\gamma}\frac{1}{N}+o\left(\dfrac{1}{N}\right)  \label{stima1}
\ee
Recalling \eqref{mumu}, we have that the eigenvalues $\mu_1$ and $\mu_2$ are either complex
conjugated, i.e. $\mu_1=\mu_2^*$, with $|\mu_1|=|\mu_2|=1$, or both real, such that $\mu_1=1/\mu_2$.
\noindent
Let us first suppose $\mu_1$ and $\mu_2$ to be complex.
Then we can write the eigenvalues in Eqs. (\ref{mu}) in the form
\be
\mu_1=e^{-i \Theta_N} \quad \text{and} \quad \mu_2=e^{i \Theta_N} \label{mubis}
\ee
where the phase $\Theta_N$ is real and is given by
\be
\Theta_N=\arctan\left(\frac{\sqrt{1-\Phi_N^2}}{\Phi_N} \right) \label{phase}
\ee
It is easy to check that
\be\label{fiN}
\Phi_N= 1-{(E-E_o)L^2\over 2N^2}+ o\left({1\over N^2}\right)
\ee
{Let} us put now
\be\label{fiN2}
\vartheta=\sqrt{E-E_o}
\ee
Note that when $E<E_o$ then, by \eqref{fiN}, we have that $\Phi_N>1$, for $N$ sufficiently large.
Therefore, since we are considering the case $|\Phi_N|\le 1$, we can assume $E\ge E_o$, i.e.  $\vartheta$  real.
Hence, the variable $\Theta_N$ has the asymptotic behaviour
\be \label{phase2}
\Theta_N=\frac{\vartheta L}{N}+o\left(\frac{1}{N}\right)
\ee
Plugging, now, \eqref{mubis} into (\ref{cheb}), one can write the Chebyschev polynomials $U_{N-1}(\Phi_N)$ as
\be
U_{N-1}(\Phi_N)=\frac{\sin(N \Theta_N)}{\sin(\Theta_N)}  \label{chebsin}
\ee
Then, from Eqs. (\ref{phase2}) and (\ref{chebsin}), one finds
\be
\lim_{N\rightarrow\infty}\left(\frac{U_{N-1}(\Phi_N)}{N}\right)^2=
\frac{\sin^2 (\vartheta L)}{(\vartheta L)^2} \label{stima2}
\ee
Let us now consider the eigenvalues $\mu_1$ and $\mu_2$ to be real, which is the case when
$|\Phi_N|>1$.  We write
\be
\mu_{1}=\sign(\Phi_N)e^{-\Xi_N} \quad \text{and} \quad \mu_{2}=\sign(\Phi_N)e^{\Xi_N}  \label{mua1}
\ee
with
\be
\Xi_N=\log|\Phi_N|+\log \left(1-\frac{\sqrt{\Phi_N^2-1}}{\Phi_N}\right)
\ee
After some algebra, one finds that
\be \label{loglog}
\Xi_N=\left\{\begin{array}{cc}
- i\frac{\vartheta L}{ N}+ o\left({1\over N}\right) &  \hbox{ if  $ \Phi_N>1$} \\ \\
i\left(\pi- \frac{\vartheta L}{ N}\right)+ \small{o}\left({1\over N}\right) & \quad \hbox{if  $ \Phi_N<-1 $}
 \end{array}\right.
\ee
Inserting Eq. \eqref{mua1} into (\ref{cheb}), one can write the Chebyschev polynomials $U_{N-1}(\Phi_N)$ as
\be
U_{N-1}(\Phi_N)= (\sign(\Phi_N))^{N-1}\frac{\sinh(N \Xi_N)}{\sinh(\Xi_N)} \label{chebsinh}
\ee
If $ \Phi_N>1$, by using  \eqref{loglog}, one obtains
\be \label{real1}
U_{N-1}(\Phi_N)= (\sign(\Phi_N))^{N-1}\frac{\sinh(- i \vartheta L)}{\sinh(- i\frac{\vartheta L}{ N})}=
(\sign(\Phi_N))^{N-1}\frac{\sin(\vartheta L)}{\sin(\frac{\vartheta L}{ N})}
\ee
Similarly, if $ \Phi_N<-1$, from  \eqref{loglog} it holds
\be  \label{real2}
U_{N-1}(\Phi_N)= (\sign(\Phi_N))^{N-1}\frac{\sinh(i(N\pi - \vartheta L))}{\sinh( i(\pi -\frac{\vartheta L}{ N}))}=(\sign(\Phi_N))^{N-1}\frac{\sin(\vartheta L)}{\sin(\frac{\vartheta L}{ N})}
\ee
Clearly, from Eqs. (\ref{real1}) and (\ref{real2}), one obtains again
\be \label{finalmente}
\lim_{N\rightarrow\infty}\left(\frac{U_{N-1}(\Phi_N)}{N}\right)^2=\frac{\sin^2 (\vartheta L)}{(\vartheta L)^2}
\ee
Finally, the claim follows from  using Eqs. (\ref{stima1}),  (\ref{stima2}) and \eqref{finalmente} in Eq. (\ref{SKP}).
This completes the proof of theorem.
\end{proof}
\br
Note that the sequence of functions $S_N$, as functions of $E$, converges pointwise but not uniformly
to $\bar{S}$ as $N\to\infty$ (see also figure \ref{S}).
\er

\br
{From Theorem \ref{teo}, the Landauer resistivity $\rho_N$, defined as}
\be
\rho_N=\frac{(1-S_N)}{S_N} \label{landauer}
\ee
also has a finite limit $\bar{\rho}$ as $N\to\infty$. Namely,
\be
{\bar{\rho}=\lim_{N\rightarrow\infty}\rho_N=\frac{1}{4E} \frac{\left(V\gamma\right)^2}{(1+\gamma)^2}
\left[\frac{\sin\left({L\sqrt{E-E_o}}\;\right)}{\sqrt{E-E_o}}\right]^2} \label{lanlim}
\ee
\noindent
{It is interesting to compare Eq. \eqref{lanlim} with the result reported in Theorem 3.1 of Ref. \cite{kam}. Observe that the factor $V\gamma$ appearing
in \eqref{lanlim} can be interpreted as the (constant) intensity of the Dirac deltas in Theorem 3.1 of  \cite{kam}, denoted therein by ``$V$'' (and corresponding to the constant $\Lambda$ defined in formula \eqref{limKP} ahead). On the other hand,} the limit $N\to \infty$ in
\cite{kam} corresponds to taking the limit $L\to \infty$ in {Eq. \eqref{lanlim}}. {Then,} one can notice the following.  {When $E<E_o$,} the resistivity
\eqref{lanlim} diverges exponentially in $L$; this corresponds to {the} item (1) of Theorem 3.1 in \cite{kam}
(i.e. $\bar\rho$ diverges exponentially in $N$ when $|\Phi|>1$). When $E=E_o$, then  \eqref{lanlim} diverges as $L^2$,
in agreement with {the} item (2)i of Theorem 3.1 in \cite{kam} (i.e. $\bar\rho$ diverges as $N^2$ when $|\Phi|=1$). Finally, when $E>E_o$,
{we have that $\bar{\rho}$} is proportional to $\sin^2(L)${, hence is O(1)}
and not converging as $L\to \infty${, as also stated in the item (2)ii of Theorem 3.1 of \cite{kam}} (i.e. $\bar\rho$ is $O(1)$  and not converging as $N\to \infty $ when $|\Phi|<1$).
\er
\begin{figure}
\centering
\includegraphics[width=7.5cm]{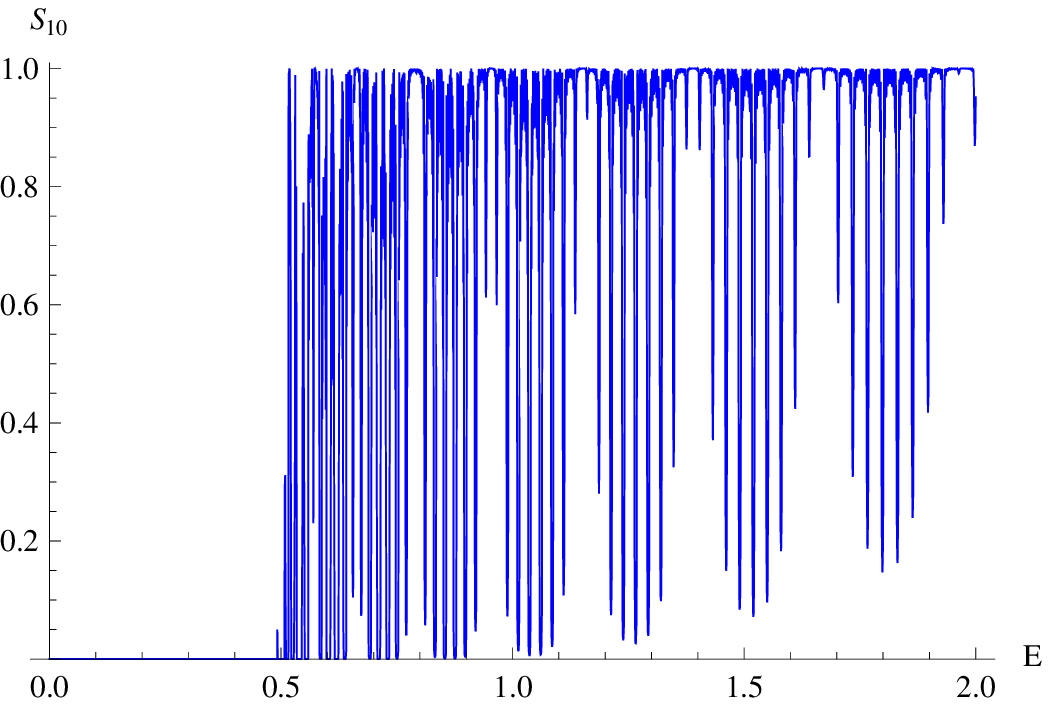}
\hspace{2mm}
\includegraphics[width=7.5cm]{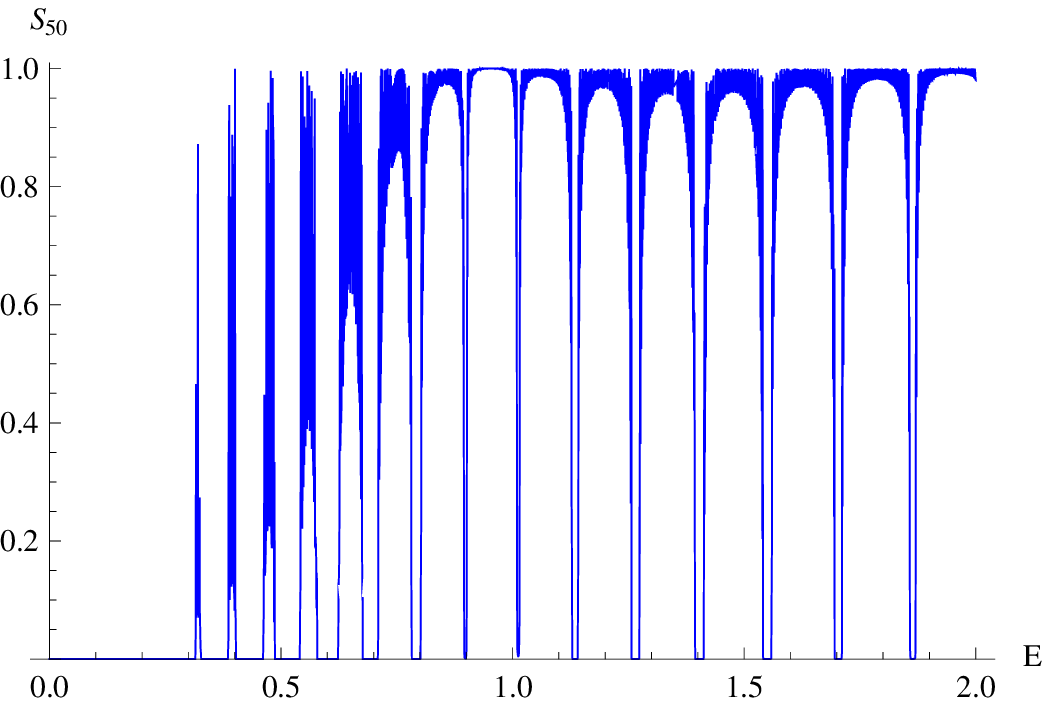}
\vspace{5mm}
\includegraphics[width=7.5cm]{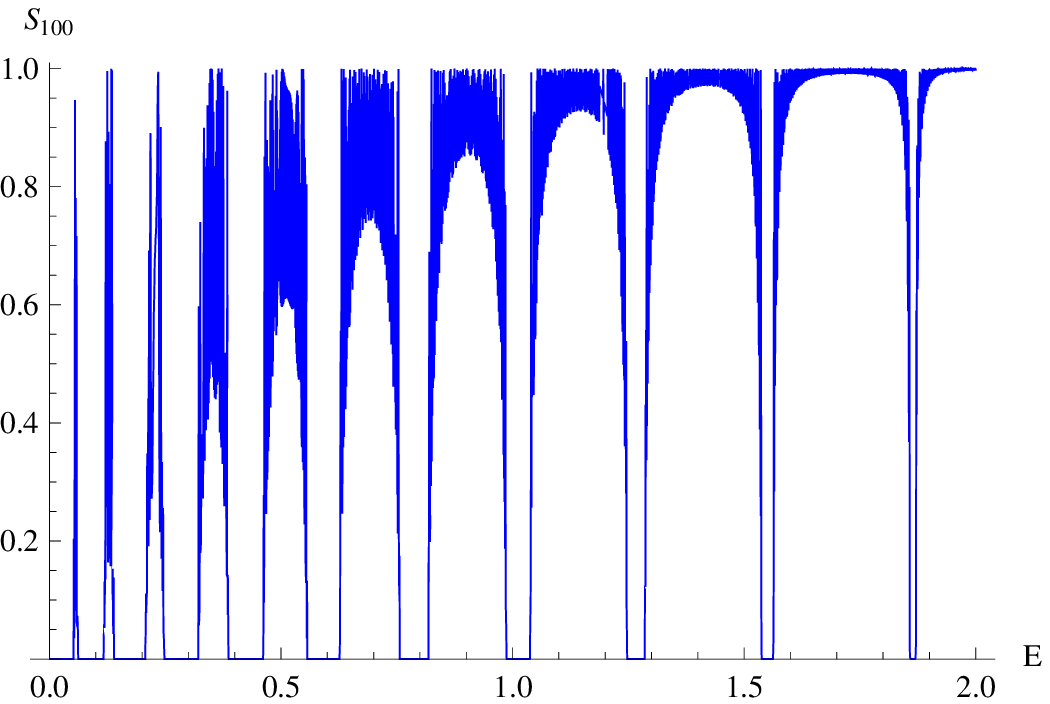}
\hspace{2mm}
\includegraphics[width=7.5cm]{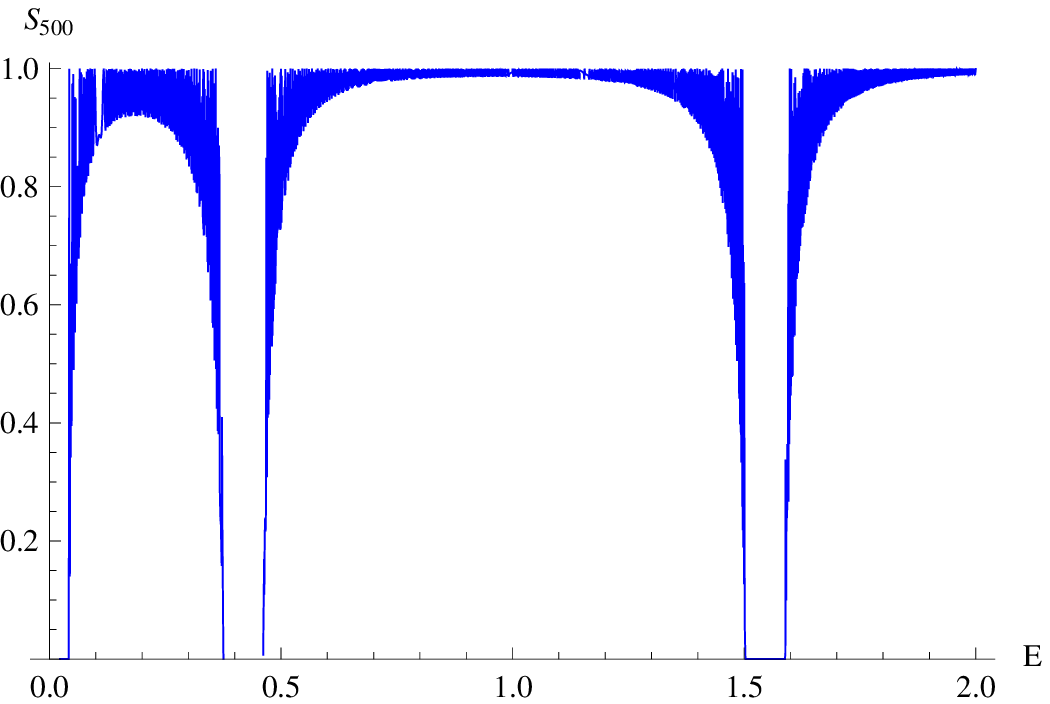}
\vspace{5mm}
\includegraphics[width=7.5cm]{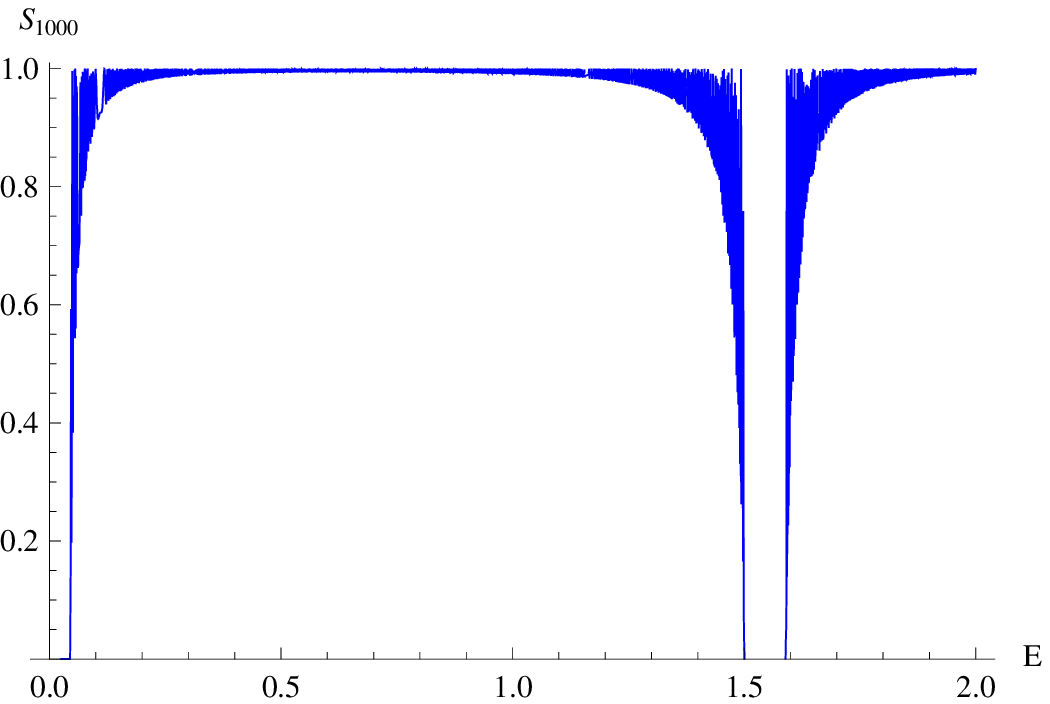}
\hspace{2mm}
\includegraphics[width=7.5cm]{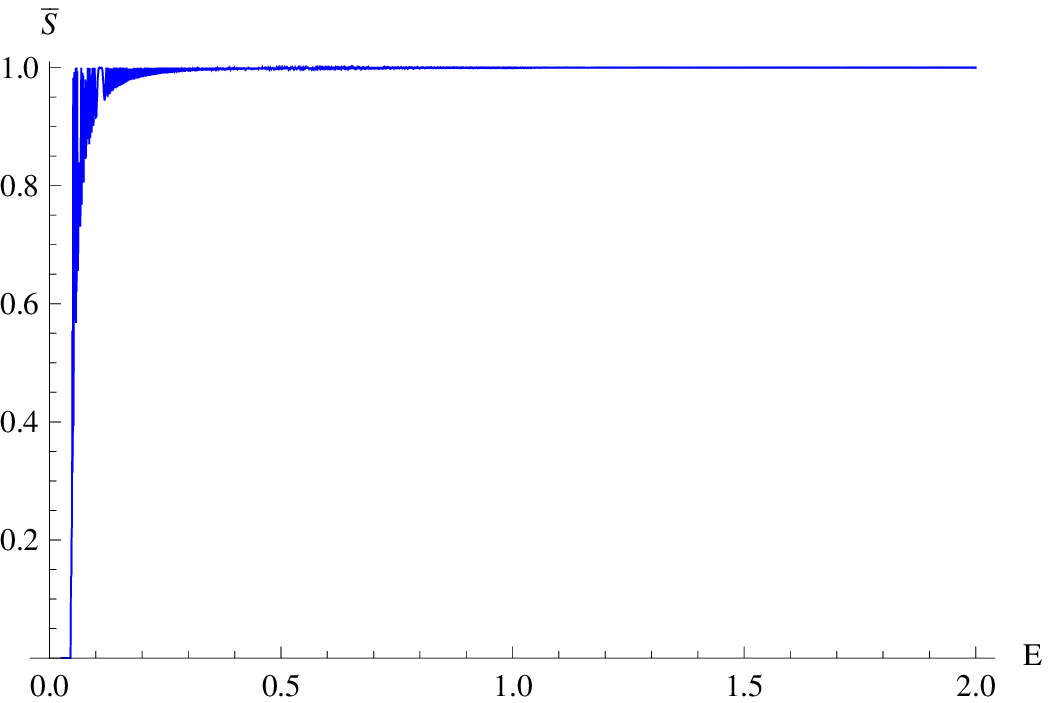}
\caption{$S_N$ vs.  $E$ for different values of $N$, with $V=0.5~eV$, $\gamma=0.1$ and $L=500~nm$. \\
{$\bar{S}$ vs. $E$ in the bottom right panel.}}\label{S}
\end{figure}

\noindent
Finally we would like to spend some words about the behavior of $S_N$,  as illustrated in Fig. \ref{S} for different values of $N$.
Indeed, Fig.  \ref{S}  shows the existence of certain energy values for which the corresponding values of $S_N$ are very small (band gaps hereafter), alternated to values for which
$S_N$ is close to 1. {The band} gaps (i.e. $S_N$ near zero) {correspond} to energy regions where $|\Phi_N|>1$; as a matter of fact, in these regions the Chebyshev polynomials are very large.
The separation between {the}
band gaps increases proportionally to $N$.
Note, {however,} that in the continuum limit {the} forbidden energy bands disappear and the transmission coefficient, {for $E\gg 1$, tends to the unity}.

The value $E_o$, defined in \eqref{Eo}, can be endowed with a physical interpretation which comes from the analysis of the dispersion
relation of the model, discussed in Sec. \ref{sec:sec4}.
Here it suffices to note that, for large $N$, the transmission coefficient starts to admit
values of order $O(1)$ only for energies exceeding a lower threshold given by $E_o+ \small{o}(1)$,
cf. Fig. \ref{S}. Such lower band gap still occurs in the continuum limit, in which case the upper bound of such
forbidden energy band is   $E_o$.

We finally remark that the explicit expression of $\bar{S}$ in Eq. (\ref{SKP2}) constitutes the analytic backbone
of the numerical results previously reported in Ref. \cite{col2013}.


\section{Dispersion relation}
\label{sec:sec4}
In this section we make a heuristic discussion about the physical interpretation of the continuum limit  treated in this paper and also determine the dispersion
relation of the model, in order to have an insight on its bulk properties. To this aim, {despite explicitly requiring $L$ to be finite, we impose the Born-von Karman periodic boundary conditions to Eq. \eqref{psiper}.
In this set up, the Bloch Theorem dictates that} the wavefunction inside the device takes{, hence,}  the form
\be
\Psi(x) = u(x) e^{i \xi x} \label{blochst}
\ee
where $\xi\in \mathbb{R}$ is the Bloch wavevector and $u(x+p)=u(x)$.\\
We now combine the continuity of the wavefunction and of its first derivative at the nodes with the  assumption \eqref{blochst}. That is, first we impose
\be\label{dr1}
\Psi\big\vert_{x=0^-}=\Psi\big\vert_{x=0^+} \quad \text{and} \quad \Psi^{'}\big\vert_{x=0^-}=\Psi^{'}\big\vert_{x=0^+}~
\ee
\be
\Psi\big\vert_{x=\lambda_N^-}=\Psi\big\vert_{x=\lambda_N^+} \quad \text{and} \quad
\Psi^{'}\big\vert_{x=\lambda_N^-}=\Psi^{'}\big\vert_{x=\lambda_N^+} \label{dr2}
\ee
and then, from  Eq. (\ref{blochst}), we may  rewrite $\Psi\big\vert_{x=\lambda_N^+}$ as
\be\label{dr3}
\Psi\big\vert_{x=\lambda_N^+}=\Psi\big\vert_{x=-\delta_N^+}~e^{i\xi p_N}
\ee
\noindent
Using the transfer matrix formalism of Sec. \ref{sec:sec2}, Eqs. (\ref{dr1}) can be cast in terms of the wave amplitudes as follows
\be
\left( \begin{array}{c}
A_{0} \\
A_{1}
\end{array} \right)=\mathbf{Q}_1  \left( \begin{array}{c}
B_{1} \\
B_{2}
\end{array} \right) \label{bl1} \quad
\ee
with
$$
\mathbf{Q}_1=\mathbf{T}^{-1}[ik,0]\mathbf{T}[z,0]
$$
On the other hand, using \eqref{dr3}, Eqs. \eqref{dr2} can be rewritten as

\be
\quad \left( \begin{array}{c}
A_{0} \\
A_{1}
\end{array} \right)=\mathbf{Q}_2  \left( \begin{array}{c}
B_{1} \\
B_{2}
\end{array} \right)  \quad
\ee
with
$$
\mathbf{Q}_2=e^{-i\xi p_N}\mathbf{\Delta}[ik\delta_N]\mathbf{T}^{-1}[ik,0]\mathbf{T}[z,0]\mathbf{\Delta}[ik\lambda_N]
$$
Hence, we get
\be
\left[\mathbf{Q}_1-\mathbf{Q}_2\right]  \left( \begin{array}{c}
B_{1} \\
B_{2}
\end{array} \right)=\left( \begin{array}{c}
0 \\
0
\end{array} \right) \label{bl3}
\ee
which admits nontrivial solutions when
\be
\det\left[\mathbf{Q}_1-\mathbf{Q}_2\right]=0  \label{det}
\ee
Equation (\ref{det}) yields the relation
\be
\cos(k\delta_N) \cosh(z\lambda_N) + \frac{z^2 - k^2}{2 k z} \sin(k\delta_N) \sinh(z\lambda_N)=\cos(\xi p_N) \label{disprelat} \quad
\ee
which may be solved to express the energy $E$ in terms of the Bloch wavevector $\xi$.\\
Figure \ref{bandsN}, obtained by numerically inverting Eq. (\ref{disprelat}), illustrates the dispersion relation $E$ vs. $\xi$ for different values of $N$: note that, as $N$ increases, the band gaps decrease and the function $E(\xi)$ approaches the continuous free particle curve.
Figure \ref{bandsN} also shows that, for any finite $N$, the band gaps occur for $\xi p_N= j \pi$, with $j\in \mathbb{N}$.

\begin{figure}[h!]
\centering
\includegraphics[width=9.1cm]{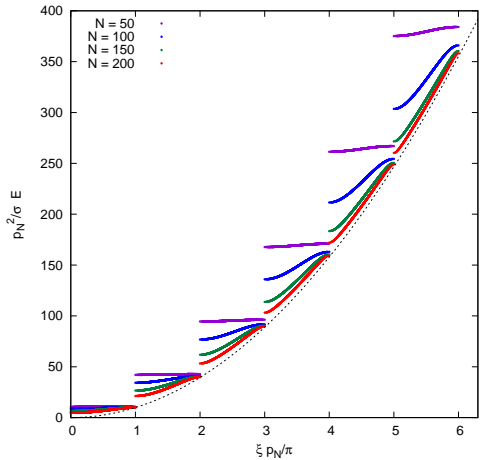}
\caption{Dispersion relation evaluated over the first six Brillouin
zones for fixed   $V=0.5$, $\gamma=0.1$ and
$L=500$  and different values of $N$. Namely,
$N=50$ (violet curve), $N=100$ (blue curve), $N=150$ (green curve),
$N=200$ (red curve) and the free particle model (black dashed line).
 }\label{bandsN}
\end{figure}
\noindent
It is worth remarking that, due to the rescaling used,  Figure \ref{bandsN} does not capture the
complete picture in the asymptotic behavior. In particular, the curves in  Figure \ref{bandsN}, as $N$ grows, tend to the free particle dispersion relation, while this is not the case for the finite K-P model, whose continuum limit admits a quadratic dispersion relation with an initial
band gap, corresponding to the energy range $[0,E_o)$. Indeed,
using  \eqref{fiN}, and that
\be
\cos(\xi p_N) = 1 - \frac{L^2\xi^2}{2 N^2}+o\left(\frac{1}{N^2}\right) \label{xilim}
\ee
one has, as $N\to \infty$, that the continuum limit dispersion relation reads
\be
E(\xi)=E_o+\hspace{0.7mm}\xi^2 \label{contlim}
\ee
Note also, from Eq. (\ref{contlim}), that in order to preserve the structure of the Bloch wavefunctions, with a real-valued wavevector $\xi$, one must require $E\geq E_o$.
No other band gaps occur {for $E\geq E_o$} in the continuum limit.\\
Thus, the resulting envelope of the periodic sequence of barriers and wells corresponds, in the continuum limit, to a single barrier of length $L$ and uniform height $E_o$, cf. Fig. \ref{barriers2}. Note, also, from Eq. \eqref{SKP2}, that for $E<E_o$ the asymptotic transmission coefficient $\bar{S}$ decreases, with $L$, as fast as $\left[\sinh\left({L\sqrt{E_o-E}}\;\right)\right]^{-2}$.

\begin{figure}[h!]
   \begin{center}
   \includegraphics[width=0.7\textwidth]{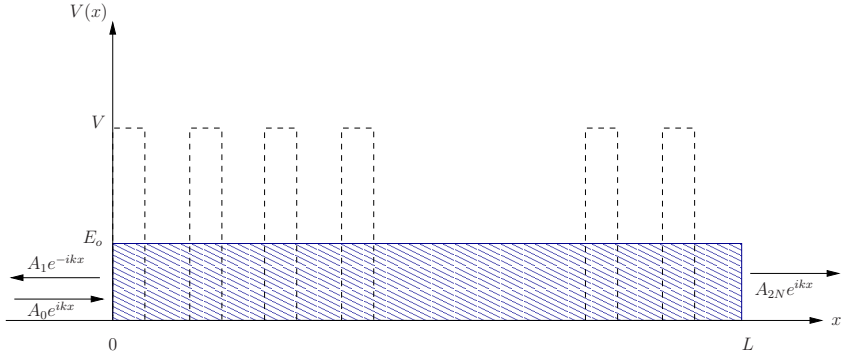}
   \caption{Continuum limit of the finite {Kronig-Penney} model: a uniform potential $E_o$ {is spread} over the domain $[0,L]$.}\label{barriers2}
   \end{center}
\end{figure}

\noindent
It is also useful to discuss how the continuum limit treated in this paper is related the limiting procedure employed, in Ref. \cite{kp}, to simplify Eq. (\ref{disprelat}) in the original K-P model. The technique considered in \cite{kp} amounts to replacing the sequence of rectangular barriers
with an array of Dirac delta functions, separated by a fixed distance $\delta$. \\
Mathematically, one considers the limits $\lambda\rightarrow 0$ and $V\rightarrow \infty$, such that the limit
\be
\lim_{\substack{\lambda\rightarrow 0\\ V\rightarrow\infty}}V \lambda = \Lambda  \label{limKP}
\ee
exists. Then, according to the limits above, the expression (\ref{disprelat}) takes the simplified structure
\be
P \frac{\sin(k\delta)}{k\delta}+\cos(k\delta)=\cos(\xi \delta) \label{DKP}
\ee
with $P=\Lambda \delta/2$.
When $P=0$, which would correspond to taking, in our model, $V=\small{o}(N)$ (i.e.,
for instance, a potential $V$ independent of $N$, as it was assumed in this work), Eq. (\ref{DKP}) gives rise to a continuous spectrum of energies (free particle regime), without band gaps. Note, however, that in the limit (\ref{limKP}) our parameter $\gamma$ {would vanish} (because in (\ref{limKP}) the width $\delta$ is kept constant, while $\lambda$ goes to zero).
In the continuum limit considered in this work, instead, the parameter $\gamma$ can take, in general, any real positive value, and is assumed to be independent of $N$.
Therefore, the effect of letting $\gamma\neq 0$, in the continuum limit of the finite K-P model, is the rise of a uniform potential $E_o$ for $x\in[0,L]$, and the presence, in the dispersion relation, of the band gap $[0,E_o)$. Note, in fact,
that the continuous spectrum, with $E_o=0$, obtained from Eq. (\ref{DKP}) with $P=0$, is recovered, in our model, for $\gamma=0$.\\

\section{Conclusions}
\label{sec:concl}

In this work we investigated the electronic transmission in a finite K-P model, by keeping the total length $L$ fixed and by varying the number $N$ of cells. We discussed, in particular, the behaviour of the transmission coefficient, of the Landauer resistivity, and of the dispersion relation in the continuum limit.\\
The analysis of such limit reveals that the particle, when letting $N$ diverge, behaves as it were subjected, along the domain $[0,L]$, to a uniform potential $E_o$. Therefore, for energy values larger than the latter threshold, the bands asymptotically disappear and the device exhibits the properties of a conductor. This asymptotic behaviour differs, in general, from that of the original K-P model, and we also clarified the connection between the two models. We stress that by introducing, in the finite K-P model, a $N$-dependence in the parameters $\gamma$ and $V$, one may access a variety of different electronic transport properties, which resembles the multitude of regimes obtained, in the original K-P model, by tuning the parameter $P$.\\
Finally, our study may offer an insight on the physics of the multi-barrier devices considered in the semiconductor technology, and is prone to be generalized also in presence of disorder \cite{izra} and external fields, by further elaborating the large deviation methods reported in Refs. \cite{col2013,colronpiz}.

\section*{Acknowledgements}
This work  has been partially supported by the Brazilian  agencies
Conselho Nacional de Desenvolvimento Cient\'{\i}fico e Tecnol\'ogico
(CNPq),  Coordena\c{c}\~ao de Aperfei\c{c}oamento de Pessoal de N\'ivel Superior (CAPES - Bolsista Jovem Talento BJT)
and  Funda{\c{c}}\~ao de Amparo \`a  Pesquisa do Estado de Minas Gerais (FAPEMIG - Programa de Pesquisador Mineiro).

\noindent
M.~C. wishes to thank Lamberto Rondoni for clarifying discussions.

\end{document}